\documentclass[english,11pt,draftcls]{IEEEtran}
\usepackage{amsfonts}                                           
\usepackage{amssymb}                                           
\usepackage{bbold}
\usepackage{amsmath}
\usepackage{amsxtra}                                            
\usepackage{amsthm}						                      
\usepackage{MnSymbol}
\usepackage{accents}
\usepackage{times}
\usepackage[english]{babel}
\usepackage{booktabs}                                               
\usepackage{textcomp}                                               
\usepackage[utf8]{inputenc}                                         
\usepackage[T1]{fontenc}

\usepackage[hyphens]{url}
\usepackage{ifpdf}
\ifpdf
\usepackage[pdftex]{graphicx}
\else
  \usepackage{graphicx}
  \usepackage{epsfig}
\fi

\usepackage{verbatim}
\ifpdf
\usepackage[pdftex,colorlinks=false,hidelinks=true]{hyperref} 
\else
\usepackage[dvips,breaklinks=true,colorlinks=false,extension=pdf]{hyperref} 
\fi

\usepackage[style=numeric,
            url=false,
            firstinits=true,
            maxbibnames=99,
            isbn=false,
            natbib=true,
            hyperref=true,
            bibencoding=utf8]{biblatex}
\AtEveryBibitem{\clearfield{note}}   

\bibliography{jabref_philipp_utf2_wjp14.bib}

\inputencoding{utf8}

\makeatletter
\let\l@ENGLISH\l@english
\makeatother

\usepackage{subcaption}
\usepackage[all]{xy}
\usepackage{xytree}
\usepackage{multirow}

\usepackage{rotating}   
\usepackage{verbatim}
\usepackage{color}
\usepackage{srcltx}                                                 
\usepackage{fancyhdr}                                              
\usepackage{enumerate}                                              
 \hypersetup {
  pdftitle          = {},
  pdfauthor 	    = {Philipp Walk},
  pdfsubject 	    = {},
  pdfkeywords 	    = {},
  hidelinks = true,
  colorlinks=false,
  pdfborder={0 0 0},
}

\newcommand{\alp}{\ensuremath{\alpha}}

\newcommand{\del}{\ensuremath{\delta}}

\newcommand{\eps}{\ensuremath{\epsilon}}

\newcommand{\gam}{\ensuremath{\gamma}}

\newcommand{\lam}{\ensuremath{\lambda}}
\newcommand{\lammin}{\ensuremath{\lambda_{\text{min}}}}
\newcommand{\rhomin}{\ensuremath{\rho_{\text{min}}}}

\newcommand{\ome}{\ensuremath{\omega}}

\newcommand{\sig}{\ensuremath{\sigma}}
\newcommand{\ta}{\ensuremath{\tilde{a}}}
\newcommand{\thh}{\ensuremath{\tilde{h}}}
\newcommand{\tg}{\ensuremath{\tilde{g}}}

\newcommand{\hx}{\ensuremath{\widehat{x}}}

\newcommand{\ts}{\ensuremath{\tilde{s}}}

\newcommand{\Omi}{{\ensuremath{\mathcal{O}}}}
\newcommand{\Odis}{{\ensuremath{\mathcal{O}_{\text{dis}}}}}

\newcommand{\cU}{{\ensuremath{\mathcal{U}}}}

\newcommand{\Stetig}{\ensuremath{ {C} }}                                          

\newcommand{\C}{{\ensuremath{\mathbb C}}}

\newcommand{\Q}{{\ensuremath{\mathbb Q}}}

\newcommand{\R}{{\ensuremath{\mathbb R}}}
\newcommand{\N}{{\ensuremath{\mathbb N}}}

\newcommand{\Z}{{\ensuremath{\mathbb Z}}}

\newcommand{\FmatrixM}{{\ensuremath{F}^{d,n}_{M}}}

\newcommand{\RA}{\ensuremath{\Rightarrow} }

\newcommand{\LRA}{\ensuremath{\Leftrightarrow} }

\newcommand{\vx}{{\ensuremath{x}}}

\newcommand{\vy}{\ensuremath{y}}

\newcommand{\hy}{{\ensuremath{\widehat{y}}}}

\newcommand{\tx}{\ensuremath{\tilde{x}}}

\newcommand{\td}{\ensuremath{\tilde{d}}}

\newcommand{\ty}{\ensuremath{\tilde{y}}}

\newcommand{\vtx}{\ensuremath{{\tilde{x}}}}
\newcommand{\vty}{\ensuremath{{\tilde{y}}}}

\newcommand{\tY}{{\ensuremath{S_{\tilde{y}}}}}
\newcommand{\tYstar}{{\ensuremath{S_{\tilde{y}}^*}}}

\newcommand{\tA}{\ensuremath{\tilde A }}                         
\newcommand{\tG}{\ensuremath{\widetilde G }}                         

\newcommand{\thmref}[1]{Satz~\ref{#1}}

\newcommand{\defref}[1]{Definition~\ref{#1}}

\newcommand{\appref}[1]{Appendix~\ref{#1}}

\newcommand{\figref}[1]{Abbildung~\ref{#1}}

\newcommand{\noi}{\noindent}

\ifx\definition\undefined
\newtheorem{definition}{Definition}         
\fi
\ifx \@definition \@empty
\newtheorem{definition}[theorem]{Definition}   
\fi

\newtheorem{corrolary}{Corrolary} 
\ifx\conjecture\undefined
\newtheorem{conjecture}{Conjecture}         
\fi
\ifx\theorem\undefined
\newtheorem{theorem}{Theorem}         
\fi
\ifx\lemma\undefined
\newtheorem{lemma}{Lemma}
\fi
\ifx\question\undefined
\newtheorem{question}{\color{red}{Question}}         
\fi
\ifx\proposition\undefined
\newtheorem{proposition}[theorem]{Proposition} 
\fi
{
\comment\noi}%
{\endcomment}

{\par\noindent{\em Beweis\/}.}%
{\hspace*{\fill}{\qed}\vspace{1ex}\par}
{\par\noindent{\em Proof\/}.}%
{\par}

{\hspace*{\fill}{}\vspace{1ex}\par}
{\par\vspace{1.5ex}\noindent{\em Remark\/}.}
{\par\vspace{1.5ex}}
\ifx\remark\undefined
\newenvironment{remark}%
{\par\vspace{1.5ex}\noindent{\em Remark\/}.}
{\par\vspace{1.5ex}}
{\par\vspace{1.5ex}\noindent{\em Example\/}. }
{\par\vspace{1.5ex}}
{\noi\vspace{0.5ex}\small}
{\vspace{0.5ex}\par\normalsize}

\newcounter{Examplecount}
{\color{gray}}
{\vspace{0.5ex}\par\normalsize}

{\renewcommand{\labelenumi}{(\roman{enumi})}\begin{list}{\labelenumi}
{\usecounter{enumi}\setlength{\labelwidth}{1.5cm}\setlength{\topsep}{0.3cm}\setlength{\itemsep}{-3pt}}}
{\end{list}}
{\renewcommand{\labelenumi}{(\arabic{enumi})}\begin{list}{\labelenumi}
{\usecounter{enumi}\setlength{\labelwidth}{1.5cm}\setlength{\topsep}{0.3cm}\setlength{\itemsep}{-3pt}}}
{\end{list}}
{\renewcommand{\labelenumi}{$\bullet$}\begin{list}{\labelenumi}
{\setlength{\labelwidth}{1.5cm}\setlength{\topsep}{0.3cm}\setlength{\itemsep}{-2pt}}}
{\end{list}}

\makeatletter
\newcommand{\set}[2]{\ensuremath{%
\setbox0=\hbox{\ensuremath{#2}}
\dimen@\ht0
\advance\dimen@ by \dp0
\left\{\left.#1\rule[-\dp0]{0pt}{\dimen@}\;\right|\;#2\right\} }}
\makeatother

\DeclareMathOperator*{\spann}{span}

\newcommand\diam{\operatorname{diam}} 
\newcommand\rk{\operatorname{rk}}

\newcommand{\Betrag}[1]{\ensuremath{ \left|#1\right| }}

\newcommand{\Norm}[1]{\ensuremath{ \left\|#1\right\| }}

\newcommand{\skprod}[1]{\ensuremath{ \left\langle #1 \right\rangle }}

\newcommand{\cc}[1]{{\ensuremath{\overline{#1}}}} 

\DeclareMathOperator{\supp}{supp}

\newcommand{\Pot}{\ensuremath{\mathcal P}}

\newcommand{\namen}[1]{{{#1}}}           

\newcommand{\shiftl}{\ensuremath{S^l}}           

\ifx \@paragraph \@empty
\makeatletter
\renewcommand\paragraph{\@startsection
{paragraph}{4}{\z@}{-3.5ex plus-1ex minus-.2ex}%
{1.3ex plus.2ex}{\normalfont\itshape}}

\fi

\DeclareUnicodeCharacter{04B4}{\CYRTETSE}
\newcommand{\Pro}{\Pi}

\renewcommand{\thmref}[1]{Theorem~\ref{#1}}

\renewcommand{\figref}[1]{Fig.~\ref{#1}}
\renewcommand{\defref}[1]{Definition~\ref{#1}}

\definecolor{gray}{rgb}{0.3,0.3,0.3}
{\color{black}}                      
{\color{black}}

\begin{document}
\onecolumn

\title{On the Stability of Sparse Convolutions}

\author{%
{Philipp Walk{\small $~^{\#}$}, Peter Jung{\small $~^{*}$}, G{\"o}tz E. Pfander{\small $~^{\dagger}$} }%
\vspace{1.6mm}\\
\fontsize{10}{10}\selectfont\itshape
$^{\#}$\,Institute of Theoretical Information Technology\\
Technical University Munich\\
Arcisstr. 21, 80333 Munich, Germany, philipp.walk@tum.de
\vspace{1.2mm}\\
\fontsize{10}{10}\selectfont\rmfamily\itshape
$^{*}$\,Heinrich-Hertz-Chair for Information Theory and Theoretical Information Technology\\
Technical University Berlin\\
Einsteinufer 25, 10587 Berlin, Germany, peter.jung@mk.tu-berlin.de
\vspace{0mm}\\
\fontsize{10}{10}\selectfont\rmfamily\itshape
$^{\dagger}$\,School of Engineering and Science\\
Jacobs University Bremen\\
Campus Ring 12, 28759 Bremen, Germany, g.pfander@jacobs-university.de%
}
\maketitle

\begin{abstract}
We give a stability result for sparse convolutions on $\ell^2(G)\times \ell^1(G)$ for torsion-free
discrete Abelian groups $G$ such as $\Z$. It turns out, that the torsion-free property prevents full cancellation in the
convolution of sparse sequences and hence allows to establish stability in each entry, that is, for any fixed entry of
the convolution the resulting linear map is injective with an universal lower
norm bound, which only depends on the support cardinalities of the sequences. This can be seen as a reverse statement of the
famous \namen{Young} inequality for sparse convolutions.  Our result hinges on a compression argument in additive set
theory. 
\end{abstract}

\section{Introduction}

Additive problems have increasingly become a focus in combinatorics, number theory, group theory and Fourier analysis as
pointed out, e.g., in the textbook of Tao and Vu \cite{Tao06}. The key hereby is an understanding of the additive
structure of finite subsets of an Abelian group $G$. The main result in the herein presented work is the application of
a recent compression result in additive set theory by Grynkiewicz \cite[Theorem 20.10]{Gry13} to sparse convolutions on
discrete Abelian groups which are \emph{torsion-free}, i.e., for any $N\in \N$ and $g\in G$ it holds $Ng=g+\dots +g=0$
if and only if $g=0$. Compressing the convolution of sparse sequences, i.e., sequences with finite support sets, reduces
to a compression of the sumset of their supports, since  $\supp(x*y)\subseteq \supp (x) + \supp(y)$. Our compression
result allows to obtain a reverse statement of Young's inequality \cite{You12,You13} for the convolution of all sparse
$(x,y)\in \ell^2(G)\times \ell^1(G)$ in form of
\begin{align}
  \Norm{x*y}_2 \geq \alp\Norm{x}_2\Norm{y}_1,\label{eq:result}
\end{align}
where $\alp>0$ exists and can be given in terms of the cardinalities of the supports of $x$ and $y$. 

For convolutions on generic locally compact Abelian groups, or LCA groups for short, such a reverse statement does not
hold in general.  To see this, take for any $N\in \N$ the finite group $\Z_N$. Full cancellation of the convolution occurs,
i.e., $x*y=0$, if we set $x=\del_0 - \del_1$ and $y=\sum_{i=0}^{N-1} \del_i$, where for each $i\in \Z_N$ the sequence
$\del_i$ is given element-wise by $\del_i(j)=1$ if $i=j$ and $0$ else. For the torsion-free group $\Z$ it is seen
easily that a finite support length or sparsity prior excludes such cancellations.  Note, that a sparsity prior on
$x$ and $y$ is not of benefit for torsion groups. E.g., for $G$ of even order $N$, we have for the $2-$sparse sequences
$x=\del_0 + \del_{N/2}$ and $y=\del_1 - \del_{N/2}$ full cancellation.%

On arbitrary LCA groups the sharp upper bounds and maximizers for Young's inequalities are known \cite{Bec75,BL76}, but
lower bounds for the reverse case have only been shown for positive functions so far \cite{Lei72,BL76,Bar98a}. On the
real line, the maximal orthogonal perturbation of a function supported on a small interval when convolved with a kernel
supported on a second small interval was analyzed in \cite{KPZ02}. 

In this work we focus on discrete LCA groups and  give lower bounds $\alp=\alp(s,f)$  for arbitrary functions
(sequences) depending only on their  discrete support length. This is a universal result and translates to
a weak stability for sparse convolutions, i.e., every $f-$sparse sequence $y$ induces a convolution map $\cdot *y$ which
is invertible over all $s-$sparse sequences $x$.  Furthermore, this allows identifiability of sufficiently sparse
auto-convolutions, since it holds
\begin{align}
  x*x -y*y =(x+y)*(x-y)
\end{align}
as shown by the first and second authors named in \cite{WJ14}.

The article is structured as follows: \thmref{thm:compression} shows  that sparse convolutions over
torsion-free discrete Abelian groups can be represented by convolutions over $\Z$ with support contained in
the first $n$ integers. The integer $n$ only depends on sparsity levels $s$ and $f$
of the convolution factors and not on the location or additive structure of the supports of the functions. 

This compressed representation guarantees  the existence of a lower norm bound $\alp(s,f)>0$ in \eqref{eq:result}.
\thmref{thm:numericbound} shows that the determination of a sharp lower bound is an NP hard problem. It can be seen as a
smallest restricted eigenvalue property of all $(s,f)-$sparse convolutions. In \thmref{thm:lowerbound} we give an
analytical lower bound $\alp(s,f,n)$, which scales exponentially in $s$ and $f$ and polynomial in $n$. Finally we will
show that indeed there exist sparse sequences, given by uniform samples of a Gaussians and a modulated Gaussians, that
show an exponential decay  of the lower bound in the sparsity as derived in \thmref{thm:lowerbound}. 

{\it Acknowledgments.} We would like to thank Holger Boche and Felix Krahmer for helpful discussions. This work was
partially supported by the DFG grant JU 2795/2 and BO 1734/24-1; it was carried in part during  G.E. Pfander's stay as
John von Neumann Visiting Professor at the Technical University Munich, the hospitality of the math department at TUM is
greatly appreciated. 

\section{Notation}

We will consider in this contribution only topological groups $G=(G,+,\Omi)$ which have a locally compact topology
$\Omi$ and are Abelian with group operation written additively, that is LCA groups. The up to normalization unique Haar
measure of the LCA group $G$ is denoted by $\mu$.  For $1\leq p<\infty$ we denote by $L^p(G,\mu)$ the Banach space of
all complex valued functions $x:G\to \C$ such that $\Norm{x}^p_p:=\int_G |x(g)|^pd\mu(g)<\infty$. The convolution of
$x,y\in L^1(G,\mu)$ is given by the formula 
\begin{align}
  (x * y)(g) =\int_G x(g)y(h-g) d\mu(g)\label{eq:convolution}
\end{align}
for $\mu-$almost every $g\in G$, see, e.g., \cite[Section 1.1]{Rud62}.  If $1\leq p,q,r \leq \infty$ and $1/p + 1/q=1 +
1/r$, then \eqref{eq:convolution} can be defined for $x\in L^p (G,\mu)$ and $y\in L^q (G,\mu)$ and \emph{Young's inequality} 
\begin{align}
  \Norm{x*y}_r \leq \Norm{x}_p \Norm{y}_q \label{eq:young}
\end{align}
holds, see for example \cite{Rud62,Bec75}. A \emph{reverse inequality} holds for \emph{positive functions} $x\in
L^p(G,\mu)$, $y\in L^q (G,\mu)$ and  $0<p, q,r\leq 1$ with $1/p + 1/q=1 + 1/r$,  indeed, we have 
\begin{align}
  \Norm{x}_p \Norm{y}_q \leq  \Norm{x* y}_r \label{eq:youngrev}
\end{align}
see \cite{Lei72},  \cite{BL76} and \cite{Bar98a}.  Here \eqref{eq:young} and \eqref{eq:youngrev} are sharp only if $p$
or $q$ is one.  For the general case, sharp bounds were obtained in \cite{Bec75,BL76}. 

In the following we consider only discrete LCA's, that is, LCA's whose topology is discrete and whose Haar measure is
therefore the counting measure. We then write $\ell^p(G):=L^p(G,\Pot(G),\lam)$ where
$\Pot(G)$ is the power set of $G$.  The set of
sparse sequences
\begin{align}
  \Stetig_c(G)=\set{x:G\to \C}{|\supp\vphantom{A} x|<\infty}\subset L^1(G,\lam) \label{eq:contcompfunc}
\end{align}
with supremum norm $\Norm{x}_\infty=\sup_{g\in G} |x(g)|$ forms a normed space, see
\appref{app:compactlysupported}, and
\begin{align}
  \int_G x(g) d\lam(g) = \sum_{g\in \supp (x) } x(g).\label{eq:finitemany}
\end{align}
The notation $C_c(G)$ for finitely supported sequences is justified since compact sets in the discrete topology are sets
with finitely many elements.    The set of $k-$sparse functions for any $k\in \N$ is then denoted by
\begin{align}
  \Sigma_k (G):=\set{\vx:G\to \C}{|\supp (x) | =\lam(\supp(x))\leq k} \subset \Stetig_c(G),
\end{align}
and we call the convolution \eqref{eq:convolution} of $x\in \Sigma_s$ and $y\in \Sigma_f$ an \emph{$(s,f)-$sparse
convolution}. Note, that the set $\Sigma_k(G)$ is neither linear nor convex.  For a set $A\subset G$ we denote the set
of functions with support contained in $A$ by $\Sigma_A (G)$. We write $\Sigma_k(\Z)=\Sigma_k$ and
$\Sigma_A(\Z)=\Sigma_A$. The first $n$ integers are denoted by $[n]:=\{0,\dots,n-1\}$ and the floor and ceiling
operation for $x\in \R$ is denoted by $\lfloor x \rfloor$ respectively $\lceil x \rceil$. By $\Sigma_s^n$ we denote all
$s-$sparse vectors $x$ in $\R^n$ respectively $\C^n$, i.e., satisfying $|\supp (x)|\leq s$.  For recovery results for
unknown $x\in \Sigma_s^n$ with $y$  not sparse and randomly  chosen, see \cite{Rauhut1,KMR13}; for the related
time-varying setting see \cite{Pfander1,Pfander2,Pfander3,KMR13}.
%

\section{Compressed Representation of Sparse Convolution} %

The support of the convolution of the functions $x$ and $y$ is contained in the Minkowski sumset $\supp (x) + \supp
(y)=I+J$ with $I,J\subset G$, see for example \cite[Theorem 1.1.6]{Rud62}.  To derive a compression result, we  map $I$
and $J$ to smaller sets that respect the additive structure of $I\cup J$.  Additive set theory provides exactly such a mapping,
which is known as a \emph{Freiman isomorphism}, see, e.g., \cite[Definition 5.21]{Tao06}.  
%
\begin{definition}[Freiman isomorphism]\label{def:freiman}
 Let $G$ and $\tG$ be Abelian groups and $A\subset G$. A  map $\phi:A\to \tG$ which satisfies
 \begin{align}
   a_1+a_2=a'_1 +a'_2 \quad\Longleftrightarrow \quad\phi(a_1)+\phi(a_2)=\phi(a'_1)+\phi(a'_2)\quad\quad,\quad a_1,a_2,a'_1,a'_2\in A \label{eq:f2isodef}
 \end{align}
 is called a \emph{Freiman isomorphism of order $2$ from $A$ into $\tG$}. 
\end{definition} 
%
\begin{remark}
  A \emph{Freiman homomorphism} $\phi_A :A\to \tG$ of order $2$  is defined analogously with $\LRA$ replaced by $\RA$ in
  \eqref{eq:f2isodef}. Further, $\phi_A$ can always be extended to the sumset $A+A$, by defining $\phi_{A+A}:A+A\to \tG$
  with $a_1+a_2 \mapsto \phi_{A+A}(a_1+a_2)=\phi_A(a_1)+\phi_A(a_2)$, see \cite[Section 2.8]{Gry13}. And vice versa, if
  $0\in A$ then $\phi_{A+A}$ allows for the definition of the map $\phi_A(a):=\phi_{A+A}(a)-a^*$ for $a\in A$ and some
  $a^*\in \tG$, where $\phi_A$ is called normalized if $\phi_A(0)=0$.  Hence we call $\phi_{A+A}$ a Freiman homomorphism
  of $A+A$ if $\phi_A$ is a Freiman homomorphism of order $2$ from $A$ into $\tG$.  If $\phi_{A+A}$ is also injective,
  then the map $\phi_A$ is a Freiman isomorphism of order $2$ in the sense of the \defref{def:freiman}.  We will omit in
  the following the subscripts and the order whenever they are directly implied by the context, see also \cite[Section
  2.8]{Gry13}.
\end{remark}
The compression quality of a Freiman isomorphism for a non-empty set $A\subset G$ depends obviously on the additive
structure of $A$ in $G$, which is characterized by the \emph{Freiman dimension} $d:=\rk(\cU(A+A))$, which is defined as
follows. The universal ambient group $G'=\cU(A+A)$ of $A+A$ is the ambient group of $G$ which contains $A+A$ such that
every Freiman homomorphism $\phi:A\to \tG$ extends to a unique group homomorphism $\phi:G' \to \tG$. Hence, an universal
ambient group also has to respect the additive structure of the elements lying outside of $A$,  see \cite[p.2 and
p.305]{Gry13} or \cite[Def 5.37]{Tao06}. The Freiman dimension $d$ of $G$ is then given by the torsion-free rank of the
universal ambient group of $G$. For example, if $G'\simeq\Z$, we have $d=1$, which is the case if $A$ is a geometric
progression, see \cite{Gry13}.  This concept goes back to Tao and Vu in \cite{Tao06}. 
By a recent result of Grynkiewicz \cite[Thm. 20.10]{Gry13} there exists a Freiman isomorphism which maps $A$ into $[n]$.
The function $n=n(|A|,d)$ is monotone increasing  in $|A|$ and $d$.  Since we have to consider arbitrary sets, we seek for each
set  a bound on the Freiman dimension $d$, which fortunately can be given in terms of the cardinality $|A|$.
The following compression result for sparse convolutions is shown below.
%
\begin{theorem}\label{thm:compression}
  Let $s,f$ be natural numbers and $G$ a discrete torsion-free Abelian group. Then for any non-zero $\vx\in \Sigma_s
  (G)$ and $\vy\in \Sigma_f (G)$  let $A\subset G$ be the union of their support sets with cardinality $m\leq s+f-1$,
  where we assume that both sets  contain zero.  Then there exists a bijective  Freiman homomorphism $\phi\colon A+A\to
  \tA+\tA \subset [2n-1]$ with $n=\lfloor 2^{2(m-\sqrt{m})\log(m-\sqrt{m})}\rfloor$ for $m\geq 5$ and $n=\lfloor
  2^{m-2}+1\rfloor$ else. Further we have 
  \begin{align}
    (\vtx*\vty)(\tg)=\begin{cases}
      (\vx * \vy) (\phi^{-1}(\tg)) &, \tg\in \tA+\tA\\
      0 &, \tg\in \Z\setminus(\tA+\tA),
    \end{cases}
  \end{align}
  where the $s-$respectively $f-$sparse sequences $\tx,\ty\in\Sigma_{\tA}$ are given with $\tA:=\phi(A)\subset[n]$ and
  $2\ta^*:= \phi(0)$  by
  \begin{align}
    \vtx(\tg)=\begin{cases}
      \vx(\phi^{-1}(\tg+\ta^*))&,\tg\in\tA\\
      0&,\tg\in \Z\setminus \tA
    \end{cases}\quad\text{and}\quad    \vty(\tg)=\begin{cases}
      \vy(\phi^{-1}(\tg+\ta^*)) &,\tg\in \tA\\
      0&,\tg\in\Z\setminus\tA
    \end{cases}\label{eq:txty}.
  \end{align}
\end{theorem}
\begin{remark}
  If $G$ is not the universal ambient group of $A+A$, then $\phi$ cannot be extended to a group isomorphism between $G$
  and $\Z$. But, since the support of $x$ and $y$  is finite, we only need a ``local'' group isomorphism on the sumset
  of the supports. Therefore, the action of the convolution in $C_c(G)$ is fully described by sparse sequences over $\Z$
  and Freiman isomorphisms.  The theorem is an extension from $G=\Z$ in \cite{JW14b} of some of the authors to arbitrary
  discrete Abelian groups like $(\Q,\Omi_{\text{dis}})$ and $(\R,\Omi_{\text{dis}})$. Note, that $\Q$ is not finitely
  generated and $\R$ is not even countable.\\ The assumption that the support sets of $x$ and $y$ contain zero is not a
  restriction since one can independently shift $x$ and $y$ such that both shifted support sets contain zero. But
  shifting $x$ and $y$ results in a shift of the convolution $x*y$ which again can be realised by a Freiman isomorphism.  
\end{remark}

\begin{proof}
The support for $x\in \Sigma_s(G)$ and $y\in\Sigma_f (G)$ is contained in some non-empty subsets
$I=\{i_0,i_1,\dots,i_{s-1}\}$ respectively   $J=\{j_0,j_1,\dots,j_{f-1}\}$ of $G$, where we can assume that $i_0=0=j_0$.
Setting $A=I\cup J$  we get $\supp(x*y) \subset A+A$ with  $|A|\leq s+f-1$.\\ 
To apply  Grynkiewicz's compression result, we need an upper bound for the Freiman dimension $d$. Note, that the Freiman
dimension can be much larger then  the  linear dimension of $\spann(G)$. Fortunately, by a result\footnote{
Note, that Tao and Vu \cite[Definition 5.40]{Tao06} define the Freiman dimension of order $2$ by $\dim(A):=d-1$.}%
of Tao and Vu \cite[Corrolary 5.42]{Tao06}  it holds $d\leq\td$ for $A\subset G$ if there exists some integer $\td\geq
1$ such that
\begin{align}
  \min(|A+A|,|A-A|) \leq (\td+1)|A| - \frac{\td(\td+1)}{2}.
\end{align}
Moreover, we know that $\min(|A+A|,|A-A|)\leq |A|^2/2 -|A|/2 +1$, given by an upper bound for the difference constant of
$A$, see e.g., \cite[pp.57]{Tao06}. Hence we aim to find $\td\geq 1$ such that
\begin{align}
  & \frac{|A|^2}{2} -\frac{|A|}{2} +1 \leq (\td+1)|A| - \frac{\td(\td+1)}{2}\label{eq:doublingtao2}.
\end{align}
This follows from
\begin{align}
 & |A|^2 - 3|A| +2 + \frac{1}{4} \leq 2 \td |A| - \td(\td+1),
\end{align}
that is, 
\begin{align}
   & \td^2 + \td(1-2|A|) + \Big(|A|^2 -3|A| +\frac{9}{4}\Big) \leq 0\label{eq:quadratic}.
\end{align}
If we have equality in \eqref{eq:quadratic} then certainly $\td$ fulfills \eqref{eq:doublingtao2}. Hence we get for the
smallest solution:
\begin{align}
  \td_- &=|A| -\frac{1}{2}-\sqrt{\frac{(1-2|A|)^2}{4} - |A|^2 +3|A| - \frac{9}{4}} =
  |A|-\frac{1}{2}-\sqrt{2(|A|-1)}.  
\end{align}
Since $\td$ should be an integer we take the ceiling respectively floor operation and obtain for every $|A|\geq 1$
\begin{align}
  d=\lceil \td_-\rceil = |A| -\lfloor \sqrt{2(|A|-1)} + 0.5\rfloor\geq 1 \label{eq:dtao}.
\end{align}
A result by  Grynkiewicz \cite[Theorem 20.10]{Gry13} implies with $A_1=A_2=A$, $A_1\cup A_2 =A$ and
$m=|A|=s+f-1$, that there exists an injective Freiman homomorphism $\psi:A+A \to \Z$ such that 
\begin{align}
  \diam \psi(A):=\max\psi(A)-\min\psi(A) \leq \bigg\lfloor d!^2 \left(\frac{3}{2}\right)^{d-1} 2^{m-2} 
  + \frac{3^{d-1} -1}{2}\bigg\rfloor=:n', \label{eq:gryn}
\end{align}
where $\psi$ is given by  a Freiman isomorphism $\psi_A$  of order $2$ for $A$, 
\begin{align}
  \psi(a_1 +a_2)=\psi_A(a_1) + \psi_A(a_2),\quad a_1,a_2\in A.
\end{align}
Since $0\in A$ we have also
\begin{align}
  \diam \psi(A)=\diam\psi(A + 0) =\diam (\psi_A (A) + \psi_A(0))=\diam \psi_A(A)\leq n'.
\end{align}
But $\psi_A$ and therefore $\psi$ are not necessarily normalized, i.e., we can not assume $\psi_A(0)=0$. Hence we search
for the smallest integer in the image:
\begin{align}
  \psi_A(a^*)=\min_{a\in A} \psi_A(a) \in \Z.
\end{align}
Then $\psi_A(a^*)+\psi_A(a^*)$ is also the smallest integer in the image $\psi(A+A)$. Now we can define
\begin{align}
  \phi_A(a):=\psi_A(a) -\psi_A(a^*),\quad a\in A
\end{align}
which is again a Freiman isomorphism of order $2$ for $A$ with the property $0\in \tA:=\phi_A(A)\subset [n'+1]$,
inducing the bijective Freiman homomorphism $\phi: A+A\to \tA+\tA$ given for $a_1,a_2\in A$ as 
\begin{align}
  \phi(a_1+a_2):=\phi_A(a_1)+\phi_A(a_2)\label{eq:bijecfrei}
\end{align}
with $0\in\tA+\tA\subset[2n'+1]$ and $\tA\subset \tA+\tA$. But note that $\phi(A)=\tA+\phi_A(0)\subset [2n'+1]$ and
therefore by bijectivity we have $A=\phi^{-1}(\tA +\phi_A(0))$ where we set $\ta^*:=\phi_A(0)=\phi(0)/2$. 
 
A worst case bound for $d\leq m-\sqrt{m}-1$ in \eqref{eq:dtao}, can be used for $m\geq 5$, see \appref{app:dbound},
which gives%
\footnote{The natural logarithm is denoted by $\ln$ and the binary by $\log$.} 
$d!\leq 2^{(m-\sqrt{m})\log(m-\sqrt{m}) - (m-\sqrt{m}-1)/\ln 2}$, see \appref{app:logfac}. With this we can bound
\eqref{eq:gryn} by
%
\begin{align}
    \diam(\tA) \leq n' & < d!^2 \left(\frac{3}{2}\right)^{m-\sqrt{m}-2}\cdot 2^{m-2} + \frac{3^{m-\sqrt{m}-2}}{2}=
    (d!^2 \cdot 2^{\sqrt{m}} +2^{-1}){3^{m-\sqrt{m}-2}}\label{eq:taofac}\\
                                &    <( 2^{2(m-\sqrt{m})\log(m-\sqrt{m}) -2(m -\sqrt{m}-1)/\ln 2  +\sqrt{m} } +2^{-1})
                                3^{m-\sqrt{m}-2}\label{eq:last2-1}
\end{align}
using the bound  $2 <2/\ln 2$, which already absorbs the $2^{-1}$ summand in \eqref{eq:last2-1}, we get
\begin{align}
  n' &< 2^{2(m-\sqrt{m})\log( m-\sqrt{m})}\cdot \underbrace{2^{-2m+3\sqrt{m} +2} 3^{m-\sqrt{m}-2}}_{=2^{-m(2-\log 3) +
    \sqrt{m}(3-\log 3)+2 -2\log 3}< 1}< \underbrace{\lfloor 2^{2(m-\sqrt{m})\log( m-\sqrt{m}}) \rfloor}_{=:n}-1 \label{eq:n2},
\end{align}
since for $m\geq 5$ the exponent is less then zero. This estimation also swallows at least an integer of magnitude $2$,
justifying the last estimation. 
For $m\in \{1,2,3,4\}$ we get the conjectured bound of Konyagin and Lev \cite{KL00} $n'\leq \lfloor 2^{m-2}\rfloor=n-1$,
which is also tight, see \appref{app:klbound}.  Let us define for each $x,y\in \Sigma_s(G)\times\Sigma_f(G)$  the
sequences $\tx,\ty\in\Sigma_{[n]}$ as in \eqref{eq:txty}, having support contained in $\tA$. Hence, $(\tx*\ty)(\tg)=0$
for $\tg\in \Z\setminus (\tA+\tA)$ and $(x*y)(g)=0$ for $g\in G \setminus(A+A)$ by \eqref{eq:bijecfrei}.  Let
$\tg=\ta_1+\ta_2\in \tA+\tA$, then $\phi^{-1}(\tg)=\phi_A^{-1}(\ta_1)+\phi_A^{-1}(\ta_2)=a_1+a_2=g\in A+A$ and we get 
\begin{align}
  (x*y)(\phi^{-1}(\tg))&\overset{\eqref{eq:convolution}}{=}
    \int_G x(h)y(g-h)d\lam(h)\overset{\eqref{eq:finitemany}}{=}\sum_{h\in A\cap(g-A)}x(h)y(g-h).
\end{align}
  If $h\in A\cap (g-A)$ then there exist $a_1, a_2\in A$ such that $h=a_1=g-a_2$. Hence the sum over $h$ is a sum over
  $a_1,a_2\in A$ which satisfy $a_1+a_2=g$. But this addition is exactly preserved under the Freiman isomorphism
  $\phi_A$.  Therefore we get \vspace{-0.1cm}
\begin{align}
  (x*y)(\phi^{-1}(\tg)) &\overset{\phantom{\eqref{eq:bijecfrei}}}{=}\sum_{\substack{a_1,a_2\in A\\ a_1+a_2=g}} 
   x(a_1) \ y(a_2)
   =\sum_{\substack{a_1,a_2\in A\\ a_1+a_2=g}}  x(\phi_A^{-1}(\phi_A(a_1)) \cdot y(\phi_A^{-1}(\phi_A(a_2))\\ 
 &\overset{{\eqref{eq:bijecfrei}}}{=}
    \sum_{\substack{a_1,a_2\in A\\ a_1+a_2=g}} x(\phi^{-1}(\phi_A(a_1) +\ta^*) \cdot y(\phi^{-1}(\phi_A(a_2)+\ta^*)\\ 
 & \overset{\eqref{eq:txty}}{=}
 \sum_{\substack{a_1,a_2\in A\\ a_1+a_2=g}} \tx(\underbrace{\phi_A(a_1)}_{=\ta_1}) \
 \ty(\underbrace{\phi_A(a_2)}_{\ta_2})
      \overset{\eqref{eq:f2isodef}}{=}\sum_{\substack{\ta_1,\ta_2\in \tA\\ \ta_1+\ta_2=\tg}}  \tx(\ta_1) \ \ty(\ta_2)\\
&\overset{\phantom{\eqref{eq:txty}}}{=}\sum_{\ta_1\in\tA}\tx(\ta_1)\ \ty(\tg-\ta_1))
   =\sum_{\thh\in\Z} \tx(\thh)\ \ty(\tg-\thh))=(\tx*\ty)(\tg)
\end{align}
\end{proof}
The smallest cardinality of $A$ is obtained if  a translation of the support of $x$ is contained in the support of $y$
or vice versa, which yields $m=\max\{|\supp(x)|,|\supp(y)|\}$. This can occur, e.g., if we consider the auto-convolution
$x=y$ or demand $x,y\in\Sigma_k$ as investigated by some of the authors in \cite{JW14b}, which yields then an upper bound for $n$ of
$\lfloor2^{2(k-\sqrt{k})\log(k-\sqrt{k})}\rfloor$.

\section{A Reverse Young inequality}

\thmref{thm:compression} shows that for each pair $(x,y)\in \Sigma_s(G) \times
\Sigma_f(G)$ there exists $(\tx,\ty)\in \Sigma^n_s \times \Sigma^n_f$ with $n=\lfloor
2^{2(s+f-1 -\sqrt{s+f-1})\log(s+f-1-\sqrt{s+f-1})}\rfloor$ such that $x$ and $\tx$, $y$ and $\ty$, and $x*y$ and
$\tx*\ty$ have the same values counting multiplicity and hence it holds for all $0<p,q,r\leq \infty$
\begin{align}
  \Norm{x * y}_{\ell^r (G)} &= \Norm{\tx * \ty}_{\ell^r_{[2n-1]}},\quad  \Norm{\vx}_{\ell^p
  (G)}=\Norm{\tx}_{\ell^p_{[n]}}\quad\text{and}\quad \Norm{y}_{\ell^q (G)}=\Norm{\ty}_{\ell^q_{[n]}}.\label{eq:normeq}
\end{align}
In fact, we can shift $x\in \Sigma_s(G)$ and $y\in \Sigma_f(G)$ such that their support sets $I$ respectively $J$
contain zero, without changing the values of the convolution.
This result allows us to derive the infimum of the norm of the convolution over all normalized $(s,f)-$sparse signals as
the minimum over all $(s,f)-$sparse signals on the spheres in $\Sigma^n_s$ and $\Sigma^n_f$, since the latter form a
finite union of compact sets. Also this allows the usage of linear algebra tools to derive an explicit norm lower bound,
i.e., a reverse Young inequality for sparse convolutions. We will show in \thmref{thm:numericbound} that each $\vty\in
\Sigma_f^n$ together with its left and right shifts generates a $(2n-1)\times n$ matrix $S_{\vty}$ whose \emph{smallest
$s-$sparse eigenvalue} \cite{RZ13} (singular value of $S_{\vty}$)
\begin{align}
  \rhomin(s,S_{\ty}) := \min_{\vtx\in \Sigma_s^n\setminus\{0\}} \frac{\Norm{S_{\ty} \vtx}_2^2}{\Norm{\vtx}_2^2}
  =\lammin(s,S^*_{\ty}S_{\ty}) \label{eq:sparseeigenvalue} 
\end{align}
provides a norm lower bound $\alp(s,f)$.  If this bound is larger zero, each sensing matrix $S_{\ty}$ can recover the
$s-$sparse signal $\tx\in \Sigma_s^n$ from $2n-1$ samples. 
\thmref{thm:numericbound} below gives an explicit lower bound on $\alp(s,f)$. 
\begin{theorem}\label{thm:numericbound}
  Let $G$ be a torsion-free discrete Abelian group and $s,f$ natural numbers. Then it holds for every $\vx\in
  \Sigma_s(G)$ and $\vy\in \Sigma_f(G)$
  \begin{align}
    \alp(s,f) \Norm{\vx}_2\Norm{\vy}_1 \leq \Norm{\vx * \vy}_2  \leq \Norm{\vx}_2\Norm{\vy}_1,
    \label{eq:dryibound}
  \end{align}
  where the lower bound satisfies 
  \begin{align}
    \alp^2(s,f)& = 
    \min_{\substack{\tx\in \Sigma^n_s, \ty\in \Sigma^n_f\\ \Norm{\tx}_2=\Norm{\vty}_1=1}} \Norm{\tx *\ty}^2_2=
    \min_{\substack{\ty\in \Sigma_f^n\\ \Norm{\ty}_1=1}}
    \rhomin (s,S_{\ty})\label{eq:nphard}
  \end{align}
  with $n=n(s,f)=\lfloor 2^{2(s+f-1-\sqrt{s+f-1})\log_2(s+f-1-\sqrt{s+f-1})}\rfloor$. The problem in \eqref{eq:nphard} is
  NP hard. %
\end{theorem}
\begin{proof}
By \thmref{thm:compression} there exists an injective Freiman homomorphism $\phi$ with images
$\tA=\phi(A)\subset[n]$ and sequences $\tx,\ty\in \Sigma_{[n]}$ satisfying \eqref{eq:normeq}. Since $\tx$ and $\ty$ are
$s-$ respectively $f-$sparse in $\Sigma_{[n]}$ we can identify them as vectors in $\C^n$.  Hence our optimization
problem reduces to
\begin{align}
 \alp^2(s,f)&=\inf_{\substack{x\in \Sigma_s, y\in \Sigma_f\\ \Norm{\vx}_2=\Norm{\vy}_1=1}} \Norm{x * y}^2_2=
 \min_{\substack{\tx\in \Sigma^n_s,\ty\in \Sigma^n_f\\\Norm{\tx}_2=\Norm{\vty}_1=1}}\Norm{\tx *\ty}^2_2.\label{eq:infmin}
\end{align}
Here the infimum is realized by a minimum, since our compression result implies that we only need to consider finitely
many support combinations, that is, $\binom{n}{s}\binom{n}{f}$ combinations. Equality in \eqref{eq:infmin} is
justified since each $(s,f)-$sparse support combination in $\Sigma^n_s\times\Sigma^n_f$ also occurs in $\Sigma_s
\times \Sigma_f$. We compute
\begin{align}
   \Norm{\tx*\ty}^2_2 &=\sum_{\tg\in \Z}\Big|\sum_{\thh\in \Z} \tx(\thh) \ty(\tg-\thh)\Big|^2\\
                      &=\sum_{\tg\in \Z} \sum_{\thh,\thh'} \tx(\thh) \cc{\tx(\thh')}\ty(\tg-\thh)\cc{\ty(\tg-\thh')}.  
\intertext{Setting  $\tg=g+\thh$ and changing the order of the sums yields} 
      &=\sum_{\thh,\thh'\in\Z}  \tx(\thh) \bigg(\sum_{g\in\Z} \ty( g+(\thh'-\thh)) \cc{\ty(g)} \bigg) \cc{\tx(\thh')}\\
      &=\sum_{\thh,\thh'\in\Z} \tx(\thh) b_{\ty}(\thh'-\thh)  \cc{\tx(\thh')}, 
\end{align}
where $b_{\ty} (k)$ is the  autocorrelation of $\ty$.  The first $n$ samples
$(b_{\ty}(0),b_{\ty}(1),\dots,b_{\ty}(n-1))$ of the autocorrelation generate the $n\times n$ positive Hermitian Toeplitz
matrix  $B_{\ty}$. Introducing the $(2n-1)\times n$ shift matrix $(\tY)_{lk}=\ty_{k-l}$ of $\ty$ for $l\in
[2n-1],k\in[n]$, which  is a  matrix containing in each row a translate of $\vty$, we have $B_{\ty}={\tYstar} \tY$. Hence
we get
\begin{align}
  \Norm{\tx*\ty}_2^2&=\skprod{\tx,B_{\ty}\tx}_{\C^n}=\skprod{\tY\vtx,\tY\vtx}_{\C^{2n-1}}=\Norm{\tY\vtx}_2^2\label{eq:bijhxhy}.
\end{align}
The minimum of \eqref{eq:bijhxhy} over all normalized $s-$sparse vectors $\tx$ is referred to the smallest
\emph{restricted eigenvalue} of $\tY$ as introduced in \cite{BRT09} or more precisely to the \emph{smallest $s-$sparse
eigenvalue} \cite{RZ13}
\begin{align}
  \rhomin (s,\tY) = \min_{\vtx\in \Sigma_s^n\setminus\{0\}} \frac{\Norm{\tY\vtx}_2^2}{\Norm{\vtx}_2^2} =
  \min_{\vtx\in \Sigma_s^n, \Norm{\vtx}_2=1} \skprod{\vtx,B_{\ty} \vtx}.
\end{align}
For each $s-$dimensional subspace in $\C^n$ and $B_{\ty}$ this is a \emph{quadratic optimization problem}.
Considering the minimum over all these quadratic optimization problems, generated by $\ty$ on the $\ell^2-$sphere of
some $f-$dimensional subspace, yields a bi-quadratic optimization problem, which was shown to be NP hard, see Theorem 2.2
in \cite{LNQY09}. Since the norms on finite dimensional spaces are equivalent, we also have NP-hardness on the
$\ell^1-$sphere:
\begin{align}
  \alp^2(s,f)=\min_{\ty\in \Sigma_f^n, \Norm{\ty}_1=1} \min_{\tx \in \Sigma_s^n, \Norm{\tx}_2=1}
  \skprod{\tx,B_{\ty}\tx}.\label{eq:biopt}
\end{align}
\end{proof}

\section{Estimation of a Lower Bound} \label{sec:lowerbound}

\thmref{thm:numericbound} shows that computing sharp lower bounds $\alp(s,f)$ is NP hard. To obtain analytic bounds, we
have to separate the bi-quadratic optimization problem \eqref{eq:biopt}. Our result is then based on independent
minimization problems over $\Sigma_s^n$ respectively $\Sigma_f^n$, where we used thhe \namen{Bernstein} inequality and an eigenvalue
estimate of Fourier minors to obtain a lower bound.  Here, the $N\times N$ Fourier matrix $F_N$  is given by
\begin{align}
  (F_N)_{lk} = e^{-2\pi i \frac{lk}{N}},\quad\quad l,k \in [N],\label{eq:fmatrix} 
\end{align}
and will serve as an approximation for the Fourier series  of $x\in \C^N$
\begin{align}
  \widehat{x}(\ome)=\sum_{k=0}^{N-1} x_k e^{-2\pi i k\ome} ,\quad\quad \ome \in[0,1).
\end{align}

%
\begin{theorem}\label{thm:lowerbound}
  Let $s,f$ and $n$ be integers, then it holds for all $x\in \Sigma_s^n$ and
  $y\in \Sigma_f^n$  %
  \begin{align}
    \alp(s,f,n) \Norm{x}_2 \Norm{y}_1 \leq \Norm{x*y}_2
  \end{align}
  with lower bound satisfying
  \begin{align}
    \alp(s,f,n) >2^{-f^2\log_2 \frac{sf}{4} +f\log_2\frac{s}{2} -\frac{3}{2}\log_2(4 f)}  n^{-f^2+f-1}.  
  \end{align}
  If $s=f$ we get the following scaling behaviour 
  \begin{align}
    \alp(s,s,n) > 2^{-(2s^2-s+1)\log\frac{s}{2}} n^{-s^2+s-\frac{1}{2}} .
  \end{align}
\end{theorem}

Before we prove this theorem, let us formulate the following universal statement, which holds for all sparse convolutions on
torsion-free discrete Abelian groups. 
%
\begin{corrolary}\label{cor:universal}
Let $G$ be a torsion-free discrete Abelian group and $s,f$ be integers, then it holds for all $x\in \Sigma_s(G)$
and $y\in \Sigma_f(G)$  %
\begin{align}
  \alp(s,f) \Norm{x}_2 \Norm{y}_1 \leq \Norm{x*y}_2
\end{align}
with lower bound satisfying
\begin{align}
  \alp(s,f) >2^{-f^2\log_2 \frac{sf}{4} +f\log_2 \frac{s}{2} -\frac{3}{2}\log_2 (4 f)}  n(s,f)^{-f^2+f-1},  
\end{align}
where $n(s,f)$ is given as in \thmref{thm:numericbound}. 
\end{corrolary}
%
\begin{proof}
By \thmref{thm:numericbound} we can restrict us to $x=\tx\in \Sigma_s^n$ and
$y=\ty\in \Sigma_f^n$ with $\Norm{x}_2=\Norm{y}_1=1$ and $\log n = -2(s+f-1-\sqrt{s+f-1})\log (s+f-1-\sqrt{s+f-1})$. %
Now we can apply \thmref{thm:lowerbound} and get the desired result.
\end{proof}

\begin{proof}[Proof of \thmref{thm:lowerbound}]
  If $s=1$ and $\Norm{x}_2=1$ we have $\Norm{x*y}_2=\Norm{y}_2$ and hence $\alp(1,f,n)=\min_{0\not=y\in
  \Sigma^n_f}\Norm{\vy}_2/\Norm{y}_1=\min_{0\not=y\in\Sigma_f^n} |\supp(y)|^{-1/2}=f^{-1/2}$. Similar, for $f=1$ we have
  $\Norm{x*y}_2=\Norm{x}_2$ and $\alp(s,1,n)=\min_{0\not=x\in\Sigma^n_s}\Norm{x}_2/\Norm{x}_2=1$. Hence we assume in the
  following $s,f\geq 2$. Then we get by the Parseval Theorem 
  \begin{align}
  \Norm{\vx * \vy}^2_2= \Norm{\hx \cdot \hy}^{2}_{L^2([0,1))}=\int_0^1 |\hx(\ome) \cdot \hy(\ome)|^2 d\ome.
  \end{align}

  Since $x$ and $y$ are vectors in dimension $n$, the absolute-square of the Fourier transforms $|\hx(\ome)|^2$ and
  $|\hy(\ome)|^2$ define real valued trigonometric polynomials of degree less than or equal to $n$, i.e.,  for
  $\ome\in[0,1)$ we get with the left--shift $\shiftl$ given by $S^l x(k)= x(k+l)$, 
\begin{align}
  p_x (\ome)&=|\hx(\ome)|^2=\Big|\sum_{k=0}^{n} x_k e^{-2\pi i k\ome}\Big|^2 \label{eq:pone} =\sum_{k,k'} x_k \cc{x}_{k'}
  e^{-2\pi i  (k-k')\ome} \\
  &=\sum_{l=-n}^{n} \underbrace{\skprod{x,\shiftl x}}_{=c_l=\cc{c}_{-l}}  e^{-2\pi i l \ome} 
  = 1 +  2\sum_{l=1}^{n} \Re(c_l) \cos(2\pi l\ome) + \Im(c_l) \sin(2\pi l \ome) \label{eq:sone}
\end{align}
where $\Re(z)$ and $\Im(z)$ denotes the real respectively imaginary part of $z\in \C$.  To estimate the maximum of the
trigonometric polynomial $p_x(\ome)$ we use the triangle inequality in \eqref{eq:pone} and the maximal support length $s$ of
$\vx$ in the Cauchy-Schwartz inequality to obtain the following bound:
\begin{align}
  \Norm{p_x}_{\infty} =\Norm{\hx}_{\infty}^2 \leq \big(\sum_{k}\chi_{\supp(x)} (k) |x_k|  \big)^2
  \leq \Norm{\chi_{\supp(x)}}_2^2 \cdot \Norm{x}_2^2 = s,
\end{align}
where $\chi_A$ is the characteristic function on $A\subset\Z$ given by $\chi_A(a)=1$ if $a\in A$ and zero else. \\
We bound the slope of $p_x$ using the \namen{Bernstein} inequality, see e.g. \cite{Zyg35}, by%
\footnote{Actually, this seems to be a rough estimation, but any other, like Riesz cosine \cite{Rie14} approximation can
not prevent the order $n$ dependence.}
\begin{align}
  \Norm{p'_x}_{\infty} \leq n\Norm{p_x}_{\infty}\leq ns. \label{eq:bernsteinx}
\end{align}
Since $p_x=|\hx|^2$  is continuous and $\Norm{\hx}^2_2=\Norm{x}^2_2=1$, there exist an
$\ome_0\in[0,1]$ such that $|\hx(\ome_0)|^2\geq 1$. 
%
\begin{figure}[ht]
  \hspace{0.4cm}
  \begin{subfigure}[b]{0.5\textwidth}
    \centering
    \includegraphics[width=\textwidth]{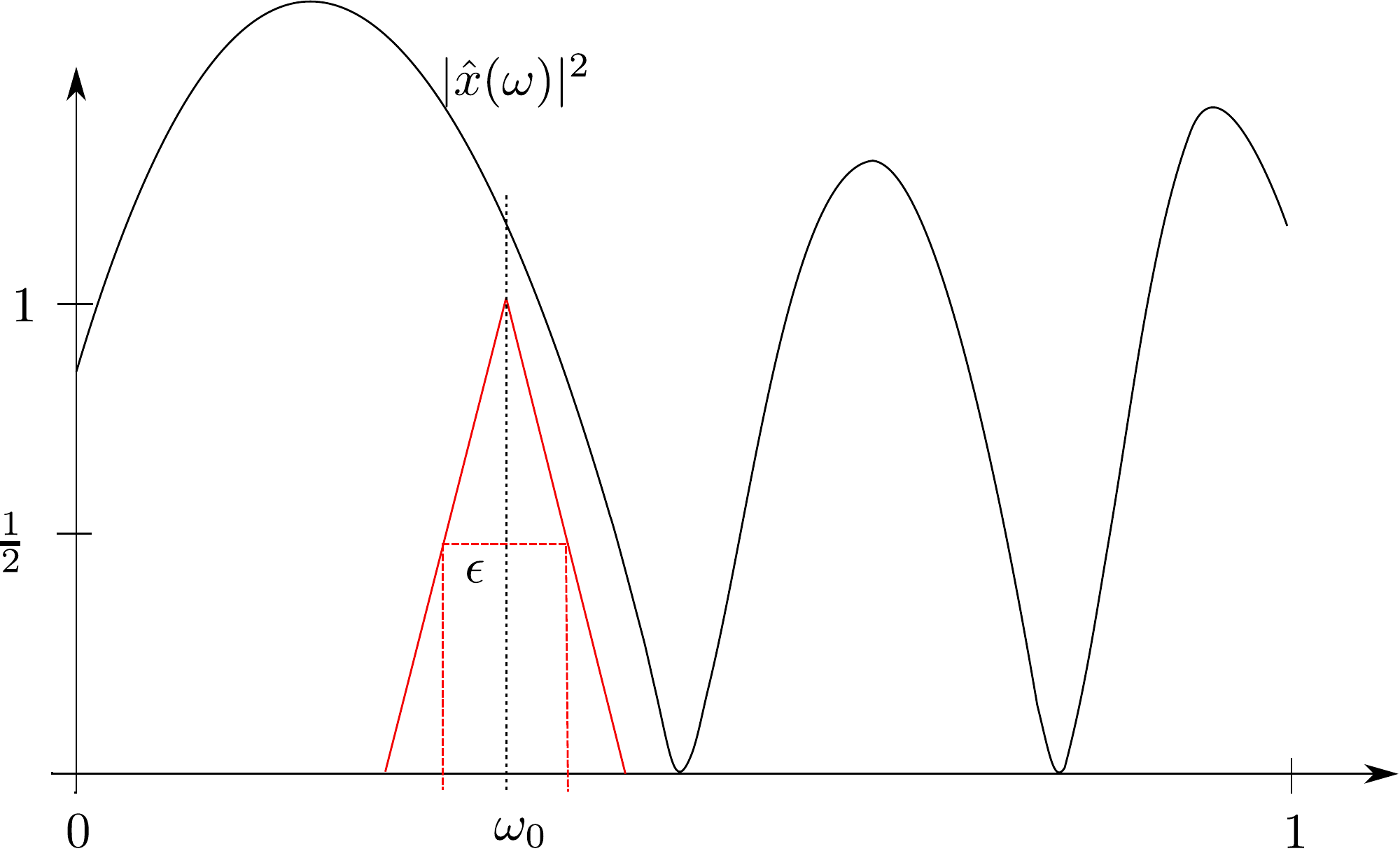}
    \caption{Bernstein inequality for $|\hat{x}|^2$}
  \end{subfigure}
  \hspace{1cm}\vspace{0.2cm}
\begin{subfigure}[b]{0.36\textwidth}
  \centering
  \includegraphics[width=\textwidth]{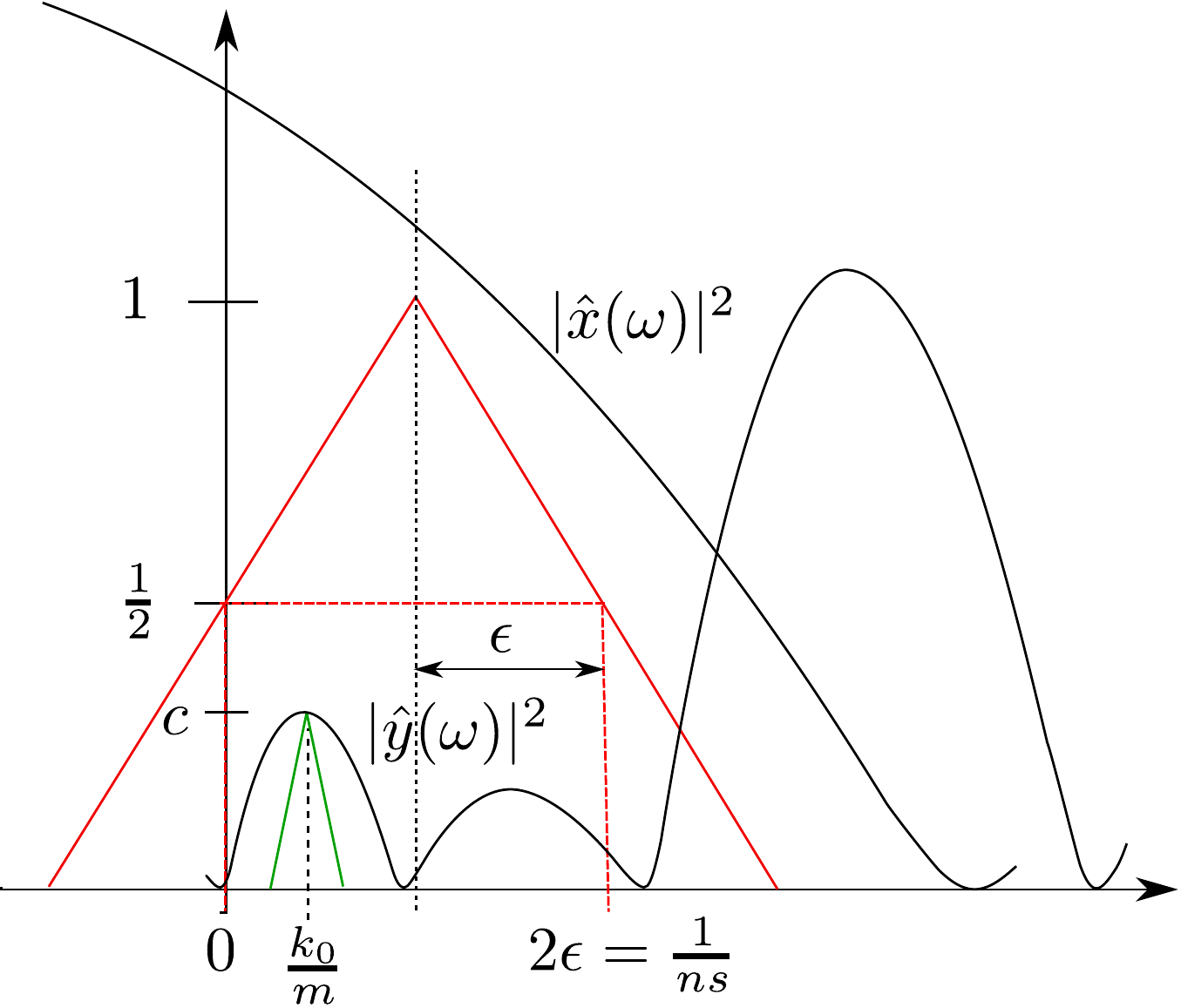}
  \caption{Bernstein inequality for $|\hat{y}|^2$}
\end{subfigure}
  \caption{Triangle lower bound using the Bernstein inequality}
  \label{fig:bernstein}
\end{figure}
%
%
Therefore we have $|\hx(\ome)|^2\geq 1/2$ in $(\ome_0 -\eps,\ome_0+\eps)$ with $\eps=1/(2ns)$, see
\figref{fig:bernstein}. Due to the invariance of the norm versus frequency shifts, we can
assume w.l.o.g. $\ome_0=\eps$ and obtain

\begin{align}
  \Norm{\vx * \vy}^2_2 = \int_0^1 |\hx(\ome) \cdot \hy(\ome)|^2 d \ome \geq \frac{1}{2} \int_{\ome_0 -\eps}^{\ome_0+\eps}
  |\hy(\ome)|^2 d\ome = \frac{1}{2} \int_0^{2\eps} |\hy(\ome)|^2 d\ome \geq 
 \frac{1}{2} \int_0^{1/ns} |\hy(\ome)|^2 d\ome .\label{eq:xbern}
\end{align}
It remains to find a lower estimate of the magnitude of $|\hy|^2$ in $[0, 2\eps)=[0,(ns)^{-1})$.

To this end, we use again a ``Bernstein triangle argument'', involving a good lower bound $c_y$ of the maximum of
$|\hy(\ome)|^2$ in $[0,(ns)^{-1})$. Indeed, since $|\hy|^2$ is a polynomial of degree less than or equal to $n$, we get
  with a similar argument as in \eqref{eq:bernsteinx}  
\begin{align}
  \Norm{p'_y}_{\infty} \leq n\Norm{p_y}_{\infty}\leq n \Norm{y}_1^2 \leq n. \label{eq:bernsteiny}
\end{align}
The area under the isosceles triangle in \figref{fig:bernstein} of height $c_\vy$ and length $2c_\vy/n$  provides  a
lower estimate for the integral in \eqref{eq:xbern}, namely
\begin{align}
  \frac{1}{2}\int_{0}^{1/ns} |\hy(\ome)|^2 d\ome >  \frac{c_\vy^2}{2n}. \label{eq:fy}
\end{align}
Let us now find a lower bound for
\begin{align}
  c_y^2=\max_{\ome\in[0,1/ns]} |\hy(\ome)|^4 = \max_{\ome\in[0,1]} \Big| \sum_k y_k e^{-2\pi i \frac{\ome}{ns}k}
  \Big|^4.\label{eq:cy2}
\end{align}
For $\ome=0$ we have 
\begin{align}
|\hy(0)|^4= |\sum_k y_k|^4,
\end{align}
which can vanish, e.g., if the support $J=\supp(\vy)$ has even cardinality and the $y_k$'s have pairwise constant
magnitude with changing sign $(-1)^k$ for $k\in[|J|]$.  Hence, for such an $\vy$ we need another frequency sample in
$[0,(ns)^{-1})$ to obtain a bound. Since $p_y$ has at most $n$ zeros in $(0,(ns)^{-1})$ an oversampling with $M=nsd$ for
some $d\in \N$ guarantees therefore non-zero entries.  Hence we get the following min-max problem:
\begin{align}
  c_y^2\geq \min_{\substack{0\not=\vy\in \Sigma_f^n\\\Norm{\vy}_1=1}}\max_{l\in [d]} {\Betrag{\sum_{k=0}^{n-1} y_k e^{-2\pi i k
    \frac{l}{M}} }^4}=
  \min_{\substack{0\not=\vy\in\Sigma^n_f\\\Norm{y}_1=1}}  {\|\FmatrixM\vy\|_\infty^4},
\end{align}
where $\FmatrixM$ is the first $d\times n$ block of the Fourier matrix \eqref{eq:fmatrix} in $M$ dimensions.
We relax this problem by considering an averaging by the $\ell^2-$norm over all $d$ samples  %
\begin{align}
  \geq\min_{0\not=\vy\in\Sigma_f^n} \frac{1}{d^2} \frac{\|\FmatrixM \vy\|^4_2}{\Norm{y}_1^4}
  \geq\min_{0\not=\vy\in\Sigma_f^n} \frac{1}{d^2|\supp y|^2} \frac{\|\FmatrixM \vy\|^4_2}{\Norm{y}_2^4}
  \geq\frac{1}{d^2f^2}\left(\min_{0\not=\vy\in\Sigma_f^n}\frac{\|\FmatrixM \vy\|^2_2}{\Norm{y}_2^2}\right)^2\label{eq:cdn}.
\end{align}
The expression in the bracket is the square of the smallest $f-$sparse eigenvalue $\lammin(f,\FmatrixM)$
\eqref{eq:sparseeigenvalue} of the $d\times n-$ Fourier minor
$\FmatrixM$. 
But the support $J=\{j_0,j_1,\dots,j_{f-1}\}\subset[n]$ of $\vy$, cuts out not more than $f$ columns of
$\FmatrixM$. Each support set $J$ with $d=|J|$ yields a $d\times d$ \namen{Vandermonde} matrix  
\begin{align}
  V_J:=
  \begin{pmatrix}
    1         & 1       & \dots & 1\\
    w^{j_{0}} & w^{j_1} & \dots & w^{j_{d-1}}\\
    \vdots & \vdots & & \vdots\\
    w^{j_{0}(d-1)} & w^{j_1 (d-1)} & \dots & w^{j_{d-1}(d-1)}
  \end{pmatrix}
 \end{align}
with $\ome=e^{-2\pi i/M}$.
Its determinant is
\begin{align}
  \det(V_J)=\Pro_{0\leq l<k<d} (e^{-2\pi i j_l/M} -e^{-2\pi i j_k/M}), 
\end{align}
which shows that $V_J$ is non-singular, see also \cite[Proposition 3.6]{KPR08}.
Moreover, by the \namen{Rayleigh-Ritz} theorem, \eqref{eq:cdn} is the smallest singular value of all Vandermonde matrices
generated by $\{w^{j_l}\}_{l=0}^{d-1}$, i.e.,
\begin{align}
  c^2=\min_{\substack{y\in\Sigma_f^n\\\Norm{y}_1=1}} 
    c_y^2\geq\min_{0\not=\vy\in\Sigma_f}\frac{1}{d^2|\supp(y)|^2}\frac{\big\|\FmatrixM \vy\big\|_2^4}{\Norm{y}_2^4} 
  \geq \min_{|J|\leq f} \min_{\substack{u\in\C^f\\\Norm{u}_2=1}} \frac{1}{|J|^4}  \big\langle u,V_{J}^*
  V_{J} u\big\rangle^2
   =\min_{|J|\leq f} \frac{\lammin^2 (V^*_J V_J)}{|J|^4}, \label{eq:sig}
 \end{align}
where the smallest eigenvalue $\lam_1=\lammin(V^*_J V_J)$ can be lower bounded by the geometric-arithmetic mean
\begin{align}
  |\det(V_J)|^2&=\det(V^*_J V_J)= \lam_1 \Pro_{j=2}^{d} \lam_j \leq \lam_1\cdot \bigg(\frac{1}{d-1} \sum_{j=2}^{d} \lam_j
  \bigg)^{d-1}\\
  &=\lam_1 \cdot \Bigg(\frac{\Norm{V_J}_F^2}{d-1}\Bigg)^{d-1}=\lam_1\cdot \Bigg(\frac{d^2}{d-1}\Bigg)^{d-1}\leq \lam_1
  (2d)^{d-1}
\end{align}
since the Frobenius norm $\Norm{V_J}_F$ of an $d\times d$ minor $V_J$ of a Fourier matrix is $d=|J|$. Hence we have 
\begin{align}
  c^2 \geq  \min_{|J|\leq f} \frac{{\lammin}^2(V^*_J V_J)}{|J|^4} \geq \min_{|J|\leq f}
  2^{-2|J|+2} |J|^{-2|J|-2} |\det(V_J)|^4 \label{eq:sig2}.
\end{align}
Taking the minimum  over all $J\subset[n]$ with $|J|\leq f$ one finds easily that $J^*=\{0,1,\dots,f-1\}=[f]$ minimizes
the determinant, since then the points $\{w^k\}_{k=0}^{f-1}$  on the unit circle are best concentrated around one.
Further, the factor in \eqref{eq:sig2} is also minimized for $|J|=d=f$. Hence we get indeed
\begin{align}
  \min_{\substack{J\subset[n]\\|J|\leq f}} |\det(V_J)|^2& =|\det(V^*_{J})|^2 =\Pro_{0\leq l<k<f} \Betrag{e^{-2\pi i l/M} -e^{-2\pi i k/M}}^2\\
  &= \Pro_{0\leq k< l <f}  4\sin^2\left(\frac{l-k}{2M}\cdot 2\pi \right)\\
  &= 2^{f(f-1)} \sin^{2(f-1)}\left(\frac{1}{M}\pi\right) \cdot
    \sin^{2(f-2)}\left(\frac{2}{M}\pi\right)  \cdot \ldots \cdot \sin^2\left(\frac{f-1}{M}\pi\right)\label{eq:f2casesin}.
\intertext{But for $s,f\geq 2$  we have $t=2(f-1)/M \in(0,1)$ and hence $\sin^2(t\pi/2)>t^2$}
& > 2^{f(f-1)} \Pro_{0< k <f} \left(\frac{2k}{M}\right)^{2(f-k)} 
  = 2^{f(f-1)}\Pro_{1\leq k <f} (f-k)^{2k} \Pro_{1\leq k < f} \left(\frac{2}{M}\right)^{2k}\label{eq:f2case}.  
\intertext{If $f\geq 3$ we have $(f-k)\geq 2$ for $1\leq k\leq f-2$ and we get}
  &\geq2^{f(f-1)} {2^{2f-4}} \left(\frac{2}{M}\right)^{2\sum_{1\leq k< f}k} =2^{f(f-1)} \cdot 2^{2f-4} \cdot \left(\frac{2}{M}\right)^{f(f-1)}
  =2^{2f^2-4}\cdot M^{-f(f-1)} \label{eq:mindet}. 
\end{align}
In fact, the case $f=2$ is also valid, since then we get the same estimate by using $2^2  \sin^2(\pi/M)\geq 2^2 (2/M)^2$
in \eqref{eq:f2casesin}.
We get with  \eqref{eq:xbern}, \eqref{eq:fy} and $f=d$ in \eqref{eq:sig2}
together with \eqref{eq:mindet} and $M=nsf$
\begin{align}
  \alp^2(s,f,n)&>\frac{1}{2n}\cdot 2^{-2f+2} \cdot f^{-2f-2} \cdot 2^{4f^2-8} M^{-2f(f-1)} \\
  &=  2^{4f^2-7-2f} (sf)^{-2f(f-1)} \cdot f^{-2f-3} \cdot n^{-2f(f-1)-1}\\
  &=  2^{4f^2-7-2f} s^{-2f^2+2f} \cdot f^{-2f^2-3} \cdot n^{-2f(f-1)-1}\label{eq:sf}.
\end{align}
Taking the square-root we have
\begin{align}
  \alp(s,f,n)&>2^{2f^2-7/2-f -f^2\log(sf) +f\log s -\frac{3}{2}\log f}  n^{-f^2+f-1}\\
   &>2^{2f^2 -f^2\log(sf) +f\log(s/2) -\frac{3}{2}\log(4 f)}  n^{-f^2+f-1}.
\end{align}
As already can be seen in \eqref{eq:mindet} we get an $-f^3\log f$ leading term in the exponent if we choose $n$ from
\thmref{thm:compression}. This can  not be further reduced in power. To see a scaling behaviour in sparsity we can
simplify for $s=f$  in \eqref{eq:sf} and obtain for the $n$ independent factor
\begin{align}
  2^{4s^2 -7 -2s} 2^{(-4s^2+2s-3)\log s} &=2^{(4s^2-2s)(1-\log s)-7-3\log s} \\
  &=2^{-(4s^2-2s)\log(s/2) - 2\log(s/2) +1} >2^{-2(2s^2-s+1)\log(s/2)}.
\end{align}
Using $1\leq \log s \leq s-1$ we get for $\alp(s)$
\begin{align}
\alp(s) &>2^{-(2s^2-s+1)\log(s/2)} n^{-s^2+s-1/2}.
\end{align}
\end{proof}

Taking $n=\lfloor 2^{2(m-\sqrt{m})\log(m-\sqrt{m})}\rfloor$ for $m=2s-1$ from \thmref{thm:compression} 
we get with \thmref{thm:lowerbound}
\begin{align}
\alp(s) >2^{-(2s^2-s+1)\log(s/2)} \cdot 2^{-2s^2(2s-1-\sqrt{2s-1})\log(2s-1 -\sqrt{2s-1})}. 
\end{align}
So even for $s=f=3$ we have an incredible low bound of $\alp(s)>2^{-37}$, which is in the order of the Planck number!
The main trouble is introduced by the lower bound of the Vandermonde determinant in \eqref{eq:mindet} which already
produces $2^{-2s^3\log s}$ if using the bound $n\sim 2^{2s\log s}$ from \thmref{thm:compression}.

\noi In the following section we see that this is not far from the truth.

\section{Cancellations for Gaussians}

In the continuous setting the Gaussian $g(t)=e^{-t^2/\sigma}$ has optimal concentration in the time-frequency plane.
This motivated us  to consider a discretized and truncated  Gaussian $g(k)$ for $k\in
\{-\ts,-(\ts-1), \dots, 0, \dots, \ts-1, \ts\}$, which gives an $2\ts+1=s$ sparse sequence for $\ts\geq1$, see \figref{fig:gausspair}.
The frequency modulation $M$ given by $(Mg)(k)=e^{\pi i k}g(k)=(-1)^k g(k)$ defines a vector with
$\widehat{(Mg)}(\gam)=\widehat{g}(\gam-\frac{1}{2})$ and has therefore the smallest overlap with $\widehat{g}$ in the Fourier domain, see
\figref{bla}. We have 
\begin{align}
  \Norm{Mg*g}_2^2 &= \sum_{l=-2\ts}^{2\ts} \Bigg| \sum_{k=\max\{-\ts,l-\ts\}}^{\min\{\ts,l+\ts\}} (Mg)(k) g(l-k)
  \Bigg|^2\\ 
  &=2 \sum_{l=1}^{2\ts} \Bigg| \sum_{k=l-\ts}^{\ts} (Mg)(k) g(l-k) \Bigg|^2 + \Bigg| \sum_{k=-\ts}^{\ts} (Mg)(k) g(-k)
  \Bigg|^2\\
  &= 2 \sum_{l=1}^{2\ts} \sum_{k,k'=l-\ts}^{\ts} (Mg)(k) (Mg)(k') g(l-k) g(l-k') + \Bigg| 1+ 2\sum_{k=1}^{\ts}
  (Mg^2)(k)\Bigg|^2\\ 
  &= 2 \sum_{l=1}^{2\ts} \sum_{k,k'=l-\ts}^{\ts} (-1)^{k+k'} e^{-\frac{k^2+k'^2}{\sig}} 
       e^{-\frac{(l-k)^2 + (l-k')^2}{\sig}} + \Bigg( 1+ 2\sum_{k=1}^{\ts}  (-1)^{k} e^{-\frac{2k^2}{\sig}}\Bigg)^2.
\end{align}
Whereas  the product of the norms is given by
\begin{align}
  \Norm{g}_2^2\Norm{Mg}_2^2= \Norm{g}_2^4= \Bigg(1 +2\sum_{k=1}^{\ts} e^{-\frac{2k^2}{\sig}}\Bigg)^2.
\end{align}
The convolved expression contains negative summands which reduce the norm $\Norm{Mg*g}_2^2$. To bound these expressions
analytically one could use bounds of the Gaussian Q-function $Q(\ts)$, see e.g. \cite{DeA09}.  

In \figref{fig:gaussmini} the bounds for the discretized Gaussian $g(k)$ and his modulated  counterpart $Mg(k)$ for
$k\in\{-(s-1)/2),\dots, (s-1)/2\}$ are plotted over various variances $\sig(s)$ and sparsity $s$ odd. It can be seen
that for $\sig(s) =(s-1)/2$ the norm of the convolution between the Gaussian pairs is minimized.  The numeric simulation
in \figref{fig:gaussop} yields for the optimal Gaussian pairs
\begin{align}
  \frac{\Norm{Mg*g}}{\Norm{Mg}\Norm{g}} \geq e^{-s/2}> 2^{-3s/4}>2^{-2s^3\log s} >\alp(s,s,s)=\alp(s).
\end{align}
This establishes an exponential decay of the lower bound in sparsity $s$. 

\begin{figure}[ht]
\hspace{-1cm}
  \begin{subfigure}[b]{\textwidth}
    \centering
     \includegraphics[width=1.1\textwidth]{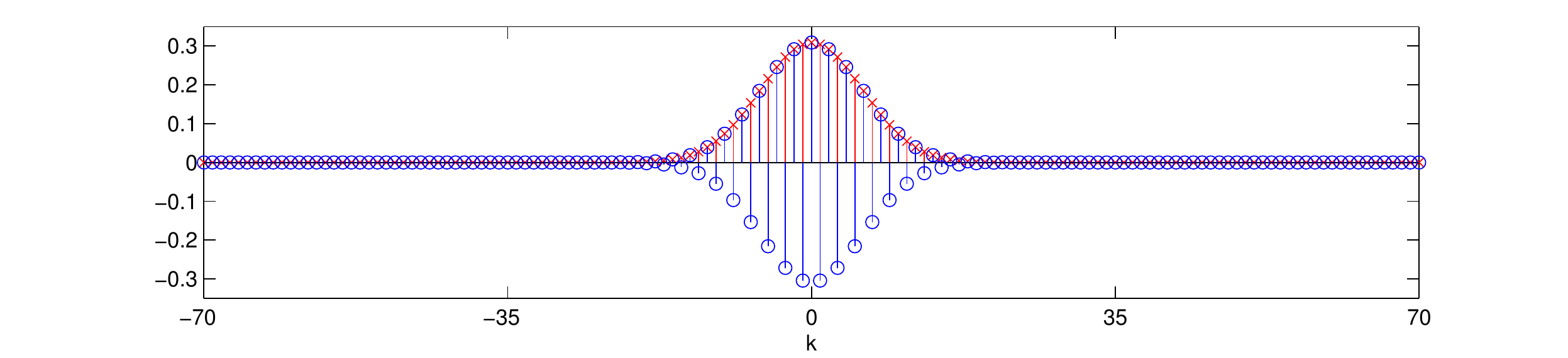}
    \caption{Normalized Gaussian pair in time domain zero padded by $s-1$ zeros.}
  \end{subfigure}\\
\hspace{-1cm}
  \begin{subfigure}[b]{\textwidth}
    \hspace{-1cm}
    \includegraphics[width=1.1\textwidth]{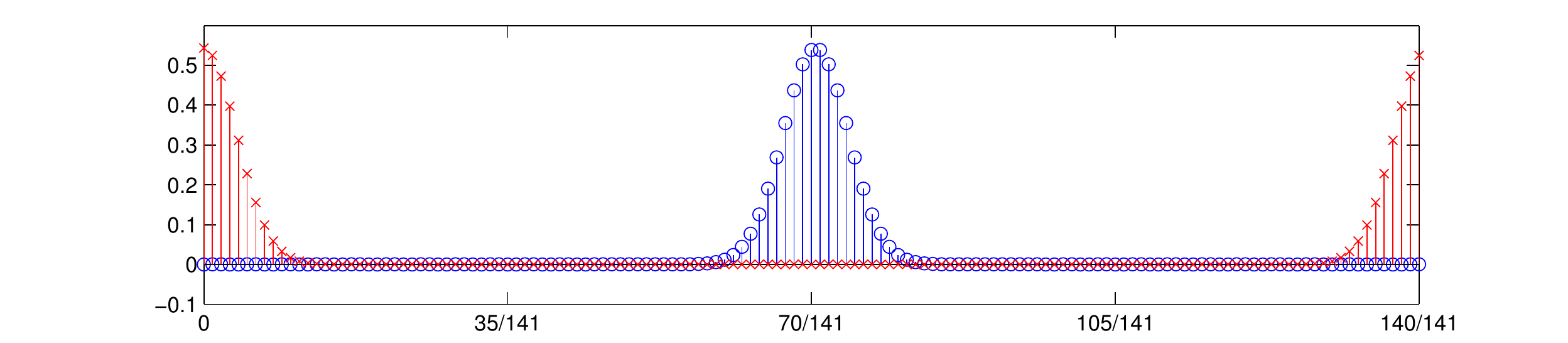}\label{fig:gaussminitime}
    \vspace{-0.3cm}
    \caption{Normalized Gaussian pair in discrete-frequency domain zero padded by $s-1$ zeros.}
  \end{subfigure}\\
  \begin{subfigure}[b]{\textwidth}
    \hspace{-1cm}
    \includegraphics[width=1.1\textwidth]{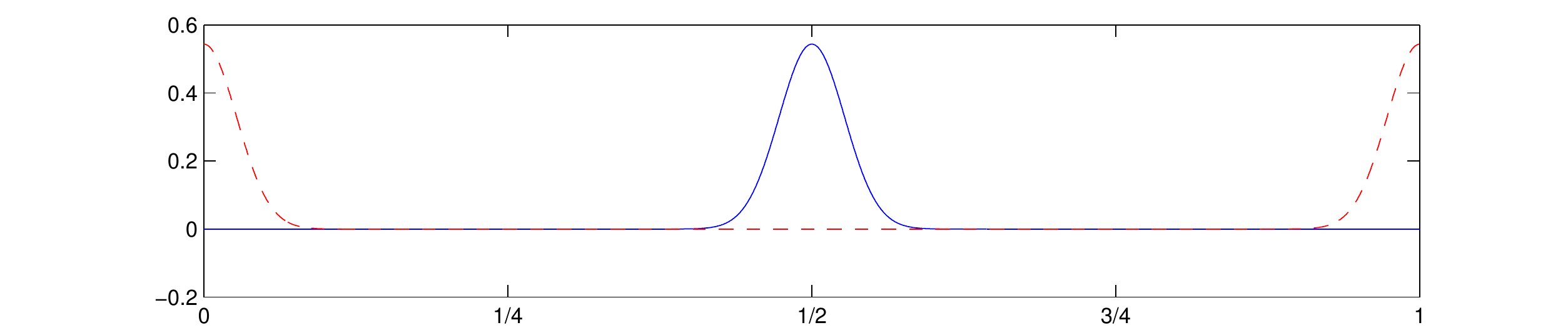}
    \vspace{-0.3cm}
    \caption{Normalized Gaussian (polynomial) pair in continuous-frequency domain.}\label{bla}
  \end{subfigure}
  \caption{Normalized Gaussian pair for $s=71$ and $\sig(s)=\frac{s-1}{2}$. Red crosses picture the Gaussian $g(k)$ and blue circles the
modulated Gaussian $Mg(k)$.}\label{fig:gausspair}
\end{figure}

\begin{figure}
  \begin{subfigure}[b]{\textwidth}
    \centering
\includegraphics[width=\textwidth]{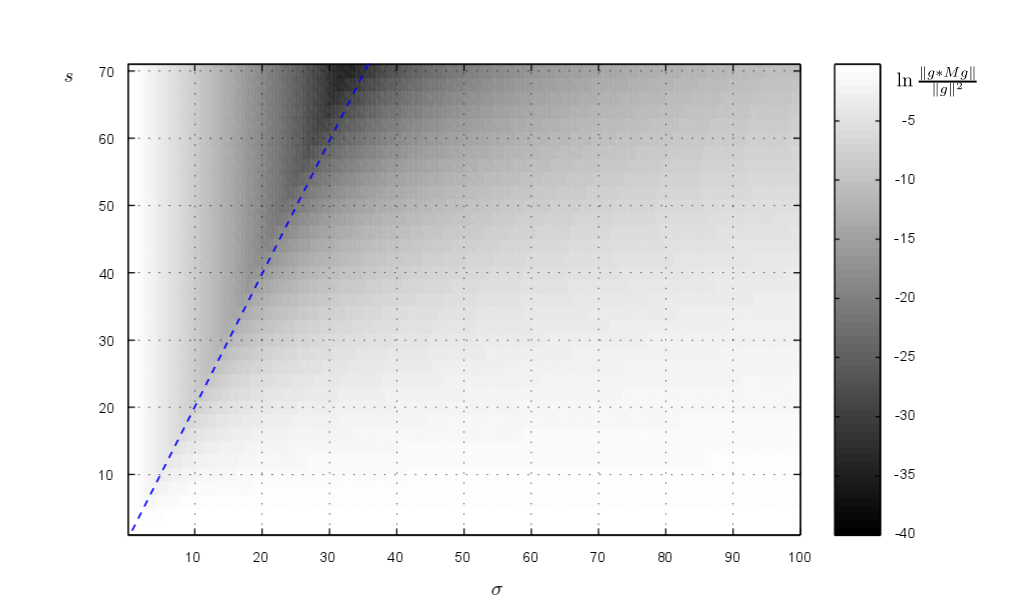}
\caption{Logarithmic lower bound over various sparsities $s$ and variances $\sigma$. The dotted line pictures $\sigma=s/2$.}
\end{subfigure}\\
\begin{subfigure}[b]{\textwidth}
  \centering
\includegraphics[width=.9\textwidth]{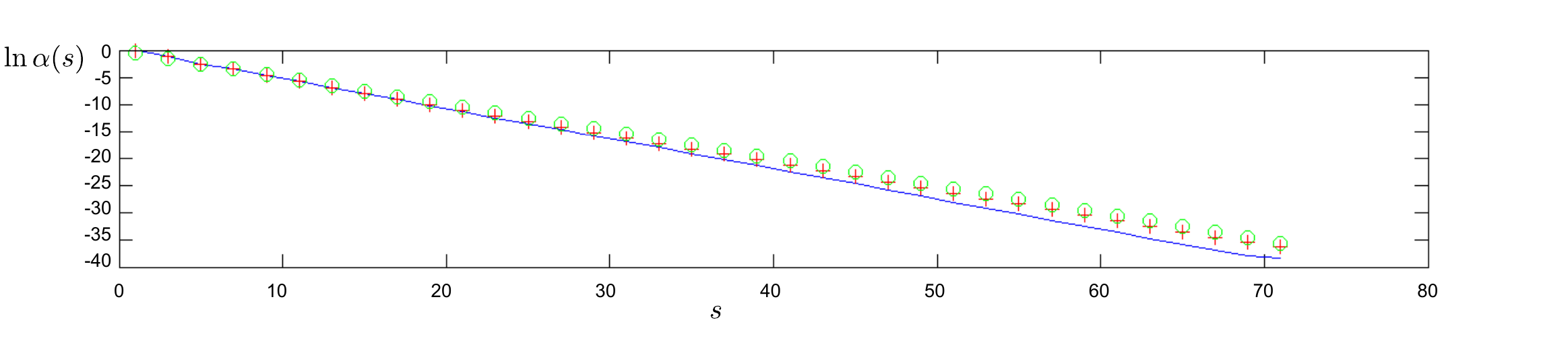}
    \caption{The blue line pictures the logarithmic $\alpha(s)$ over sparsity $s$ and best variance (numerically evaluated). Red crosses are the Gaussian pairs with
      $\sig=\frac{s-1}{2}$. Green circles pictures the bound $e^{-s/2}$.}
\label{fig:gaussop}
\end{subfigure}
\caption{Optimal variance $\sig(s)$ for the Gaussian pairs over various values of the sparsity $s$.}
  \label{fig:gaussmini}
\end{figure}




\newpage \appendix 

 \subsection{Compactly Supported Functions on  Groups with a Discrete Topology}\label{app:compactlysupported}
 Let $G=(G,+,\Pot(G))$ be a group equipped with an additive group operation $+$ and the discrete topology $\Pot(G)$ as
 topology such that $G$ is a discrete group. Further, we equip $G$ with the counting measure $\lam$, then $G=(G,\lam)$
 is a measure space.  Then we write $\ell^p(G)=L^p(G,\Pot(G),\lam)$, see e.g. \cite[Remark 10.7]{LM05}.

 Let us consider any measure $\mu$ and topology $\Omi$  on $\C$. The set of continuous complex functions
 \begin{align}
   C(G):=\set{x:G\to \C}{x \text{ continuous}}
  \end{align}
  is equal to the set of all maps $x: G\to
 \C$, since it holds for any open set $V\subset \C$
 \begin{align}
   x^{-1}(V) \subset G.
\end{align}
Therefore $x^{-1}(V)\in \Pot(G)$, i.e., an open set, and hence $x$ is continuous, see e.g. \cite[Def.1.2]{Rud87}.  Also,
the space of continuous functions with compact support
\begin{align}
  \Stetig_c(G)=\set{x:G\to \C}{ x \text{ continuous}, \cc{\supp(x)} \text{ compact}} \label{eq:contcomp}
\end{align}
with $\supp(x):=\{x\in G\mid x(g)\not= 0\}$,
equal the space of all finitely supported sequences \eqref{eq:contcompfunc}. 
If $G=\Z^d$ for some $0\not=d\in \N$, i.e., the
group is countable and finitely generated, then we get $\Norm{x}_p^p=\sum_{g\in \Z^d} |x(g)|^p$ for $0 < p<\infty$ and
$\Norm{x}_\infty=\sup_{g\in G} |x(g)|$. The convolution in \eqref{eq:convolution} becomes the \emph{discrete
convolution} for $\vx,\vy\in \ell^1(\Z^d)$  
\begin{align}
  (\vx * \vy)(g)= \sum_{h\in \Z^d} x(h)y(h-g) \label{eq:defconv},
\end{align}
see, e.g., \cite{Rud62}. Convolution turns $\ell^1 (G)$ into a unital Banach algebra and
$\Stetig_c(G)$ into a unital normed algebra both with unit $\del_0$, given component wise for $0\not=g\in G$ by
$\del_0(g)=0$ and $\del_0(0)=1$. Note, $L^1(G,\mu)$ and $\Stetig_c(G)$ have only a unit if $G$ is discrete, see, e.g.,
\cite{Que01}.

\subsection{Upper Bound for the Freiman Dimension}\label{app:dbound}
  
By a simple estimation of Tao and Vu \eqref{eq:doublingtao2} one can show for the Freiman dimension $d$ of $A$ with
$|A|=m$ the bound \eqref{eq:dtao} as
\begin{align}
  d \leq \td \leq d_{TV}:= m - \lfloor \sqrt{2m-2}+\frac{1}{2}\rfloor\label{eq:taovuestimate}.
\end{align}
It is easy to verify, that for $m\in\{1,2,3\}$ we get $d=1=d_{TV}$. For $m=4$ we have also $d=1=\td$ as solution of the exact
Tao and Vu problem \eqref{eq:doublingtao2}, since
\begin{align}
  & \frac{|A|^2}{2} -\frac{|A|}{2} +1 =8-2 +1 =7\leq 7=2\cdot 4 - 1=(\td+1)|A| - \frac{\td(\td+1)}{2}.
\end{align} 
Hence we have only to consider $m\geq 5$ in \eqref{eq:taovuestimate}. We can show
\begin{align}
  \lfloor \sqrt{2m-2}+\frac{1}{2}\rfloor \geq\sqrt{2m-2}-\frac{1}{2}  \geq \sqrt{m} +1,
\end{align}
which is equivalent to 
\begin{align}
    \sqrt{2m-2} \geq \sqrt{m} + \frac{1}{2}.
  \end{align}
By squaring both sides we get indeed
\begin{align}
  \quad 2m-2  \geq m  +\sqrt{m} +\frac{1}{4}  \quad \LRA \quad m  \geq \sqrt{m} +2 + 0.25,
\end{align} 
which is true for all $m\geq 5$. Hence we have for $m\geq 5$
\begin{align}
  d_{TV}&\leq m -\sqrt{m}-1 \label{eq:bound5}.
\end{align}
Note, the bounds  getting worse for $m$ getting large, see \figref{fig:nbound}.

 \subsection{Log Factorial Bounds}\label{app:logfac}

  For $d\geq 1$ we can write the natural logarithm of the factorial as
  \begin{align}
    \ln(d!)=\ln(\Pro_{x=1}^{d} x)= \sum_{x=1}^d \ln x
  \end{align}
  which can be upper bounded by the integral
  \begin{align}
    \ln(d!) & \leq \int_{0}^{d} \ln(x+1) dx\\
            & =    (x+1)\ln(x+1) -x\big\mid_{0}^{d} \\
            & = (d+1)\ln(d+1) -d.
  \end{align}
  Hence, for the binary logarithm $\log$, we get
  \begin{align}
    \log (d!)\leq ((d+1)\ln(d+1) -d)/\ln 2=(d+1)\log(d+1) -d/\ln 2.
  \end{align}
  Now we see that also for $m=4$  we get with \eqref{eq:bound5} 
  \begin{align}
    d! \leq 2^{(m-\sqrt{m} )\log(m-\sqrt{m}) -(m-\sqrt{m}-1)/\ln 2 }\label{eq:dfacbound}
  \end{align}
  since for $m=4$ we can show that the Tao-Vu bound gives $d_{TV}=1$ in \eqref{eq:doublingtao2} and therefore again the
  Konyagin-Lev bound with $n(4)=2^{4-2}=4$.  Hence $d!=1$. The right hand side of \eqref{eq:dfacbound}
  \begin{align}
    2^{(4-2)\log(4-2)-1/\ln 2 } > 2^{0.5} > 1.
  \end{align}

\subsection{Konyagin-Lev Bound}\label{app:klbound}

Konyagin and Lev conjectured the following in \cite{KL00}.
\begin{conjecture}
  For a finite set $A\subset \Z$ with $|A|=m$ there exists a Freiman isomorphism $\phi:A\to \Z$ such that
  $\phi(A)\subset [2^{m-2}+1]$.
\end{conjecture}
The bound $2^{m-2}$ by Konyagin-Lev is given by the black line in \figref{fig:nbound} and is sharp for Sidon sets such
as the geometric progression $\{0,2^0,2^1,\dots,2^{m-2}\}$.  The red bound is the one derived in
\thmref{thm:compression}.  The blue bounds are the one derived from Grynkiewicz bounds with a sharp version (floor
operation) of the Tao-Vu bound \eqref{eq:taofac} with the exact factorial expression. The dotted blue line is the
Grynkiewicz bound \eqref{eq:gryn} combined with Tao-Vu's bound on $\td$ gives in \eqref{eq:dtao}. The first explicit
bound given by one of the author in \cite[Theorem 2]{Wal14} as $n(m)=10^{10^8(s+f-1)\log^3(s+f-1)}$ was a rough bound
even in case of $s+f-1=m$, and is therefore not plotted in \figref{fig:nbound}. However, the proof idea of \cite[Theorem
2]{Wal14}  is very similar to \cite[Theorem 20.10]{Gry13}. 
\begin{figure}[ht]
  \centering
  \vspace{-1cm}\hspace{-0.85cm}\includegraphics[scale=0.65]{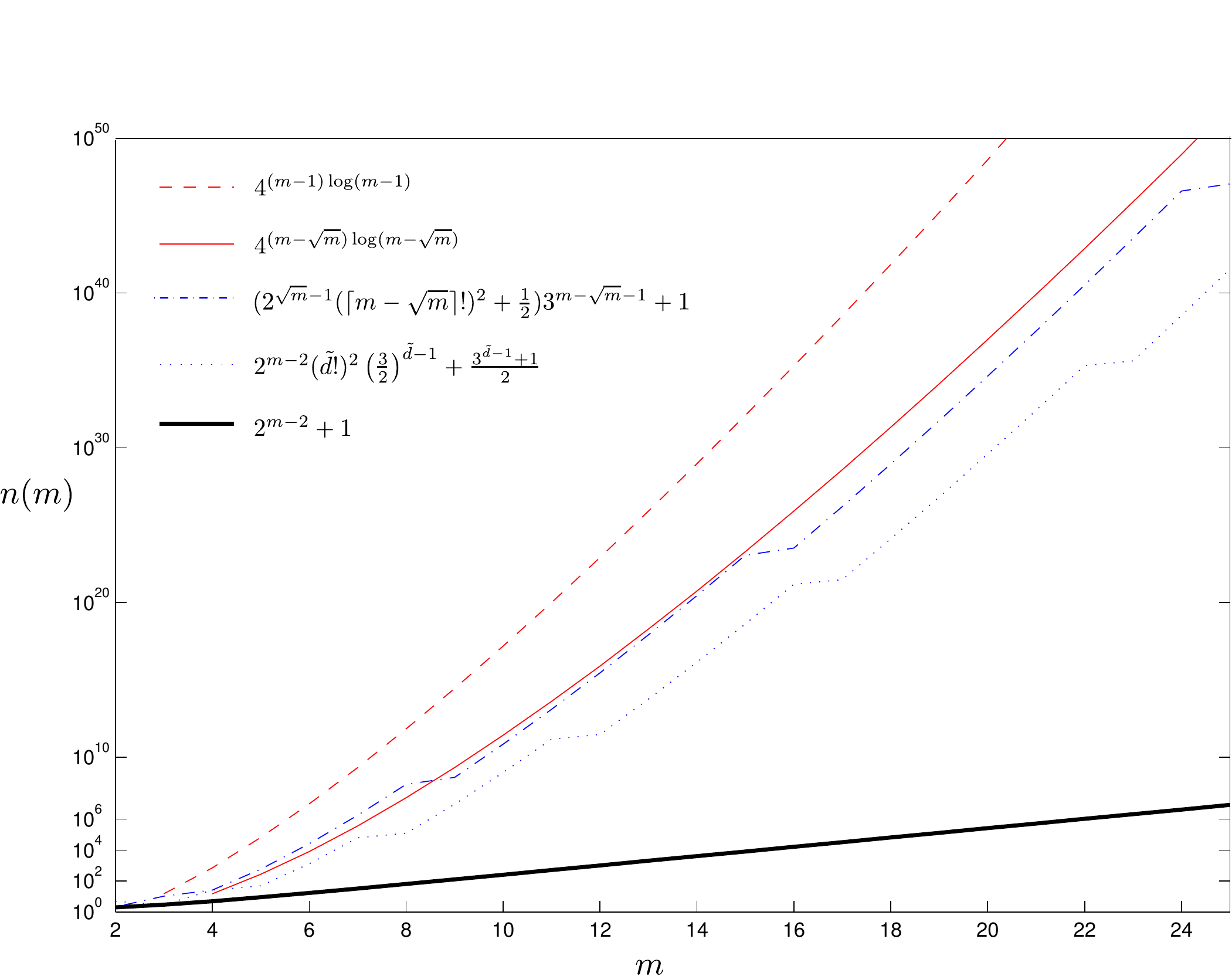}
   \caption{Analytic dimension bounds for $n=n(m)$. The black curve is the Konyagin-Lev conjecture, whereas the blue
   lines correspond to exact expression of a Grynkiewicz result combined with the Tao-Vu bound on $d$. The red lines are the
   simplified upper bounds, where the straight red bound is used in \thmref{thm:compression}. }
   \label{fig:nbound}
\end{figure}
\newpage
\printbibliography

\end{document}